%% file: TWC_2014_final.tex
\documentclass[10pt,final,letterpaper,twocolumns,romanappendices]{IEEEtran}
\usepackage{bm,cite}
\usepackage{amsmath}
\usepackage{amsthm}
\usepackage{amsbsy}
\usepackage{epsfig}
\usepackage{latexsym}
\usepackage{amssymb}
\usepackage{wasysym}
\usepackage{placeins}
\usepackage{dsfont}
\usepackage{setspace}
\usepackage{hyperref}
\newtheorem{theorem}{Theorem}
\newtheorem{definition}{Definition}

\newtheorem{corollary}{Corollary}

\input{Definitions}
\begin{document} 
\title{Interference Alignment with Incomplete CSIT Sharing}
\author{Paul de Kerret and David Gesbert\\Mobile Communications Department, Eurecom\\
Campus SophiaTech, 450 Route des Chappes, 06410 Biot, France\\\{dekerret,gesbert\}@eurecom.fr}

\maketitle


\begin{abstract}
In this work\footnote{This work has been performed in the framework of the European research project SHARING, which is partly funded by the European Union under its FP7 ICT Objective 1.1. Some preliminary results have been published in \cite{dekerret2012_ISWCS}.} we study the impact of having only \emph{incomplete} channel state information at the transmitters (CSIT) over the feasibility of interference alignment (IA) in a $K$-user MIMO interference channel (IC). Incompleteness of CSIT refers to the perfect knowledge at each transmitter (TX) of only a sub-matrix of the global channel matrix, where the sub-matrix is specific to each TX. This paper investigates the notion of IA feasibility for CSIT configurations being as incomplete as possible, as this leads to feedback overhead reductions in practice. We distinguish between antenna configurations where (i) removing a single antenna makes IA unfeasible, referred to as \emph{tightly-feasible} settings, and (ii) cases where extra antennas are available, referred to as \emph{super-feasible} settings. We show conditions for which IA is feasible in strictly incomplete CSIT scenarios, \emph{even in tightly-feasible settings}. For such cases, we provide a CSIT allocation policy preserving IA feasibility while reducing significantly the amount of CSIT required. For super-feasible settings, we develop a heuristic CSIT allocation algorithm which exploits the additional antennas to further reduce the size of the CSIT allocation. As a byproduct of our approach, a simple and intuitive algorithm for testing feasibility of single stream IA is provided.
\end{abstract}

\IEEEpeerreviewmaketitle


\begin{keywords}
Interference alignment, Interference Channel, Channel State Information, Degrees-of-Freedom 
\end{keywords}

\section{Introduction}\label{se:Intro}
 
Although multi-transmitter coordinated transmission such as interference alignment (IA) \cite{MaddahAli2008,Cadambe2008} constitutes a promising tool to combat interference, coordination benefits come at the expense of acquiring accurate enough channel state information (CSI) at the transmitters (TXs) and sharing it across all TXs whether explicitly or implicitly\cite{Gesbert2010}. In the case of multi-antenna based IA without channel extension (i.e., when the transmission schemes are not spread across multiple time slots), which is the focus of this work, some form of CSI at the TXs (CSIT) is required to compute the precoders at each one of the TXs and can result in a significant overhead in practice.

The IA literature for static MIMO channels is rich in methods improving the efficiency of the precoding schemes at finite SNR and reducing the complexity of the algorithms\cite{Gomadam2008,Schmidt2009,Kumar2010,Santamaria2010,Papailiopoulos2010,Peters2012}. Yet, obtaining the CSIT at the TXs represents a major obstacle to their practical use \cite{Garcia2011}. Thus, the study of how CSIT requirements can somehow be alleviated has become an active research topic in its own right~\cite{Gomadam2008,Tresch2009,Thukral2009,Krishnamachari2010,Schmidt2009,Shi2010,ElAyach2012a,ElAyach2012b,Lee2010,Schreck2013}. Several approaches have been proposed in this direction and are briefly summarized below. 

One strategy consists in developing iterative methods that can exploit local measurements made by the TXs on the reverse link or progressive feedback mechanisms \cite{Gomadam2008,Shi2010,Schmidt2009}. Such methods rely on the fact that, through iterations, enough CSIT is acquired to allow convergence in a distributed manner toward a global IA solution. In \cite{Cho2012}, the amount of information exchanged between the TXs is reduced by letting some TXs compute their precoder and share it instead of sharing the CSI. Yet, this is obtained at the cost of an increased delay because the improvement is obtained by letting the TXs successively compute their precoders. Also, the scheme described in \cite{Cho2012} is only applicable in some particular antenna configurations. In \cite{Suh2008,Suh2011}, IA is adapted to the configurations of cellular networks. In \cite{Lee2010}, multi-user diversity is exploited to obtain approximately aligned interference without the requirement of full CSIT. In \cite{Peters2012}, the trade-off between serving jointly all the users via IA in a large IC or serving the users orthogonally in different frequency bands is investigated. To reduce the overhead due to the CSI feedback, an intermediate solution is found where the IC is split into smaller ICs to improve the overall efficiency of the transmission scheme. In \cite{Thukral2009,Krishnamachari2010}, the number of CSI quantization bits that is sufficient to achieve the optimal number of degrees of freedom (DoF) under IA is provided.
 
Another major question regarding IA in static MIMO channels is the \emph{feasibility problem}, i.e., to determine whether the antenna configuration at the TXs and the receivers (RXs) allows the interference-free transmission of all data streams to all users. This problem was first investigated in \cite{Yetis2010} by counting the number of variables available for beamforming and the number of IA equations to obtain necessary conditions for IA feasibility. Since then, the understanding of this problem has significantly improved thanks to the use of algebraic geometry: A necessary and sufficient condition for IA feasibility when all TXs and RXs have the same number of antennas is given in \cite{Bresler2011b} and is extended in \cite{Razaviyayn2012,Gonzalez2012b,Ruan2013} to the general case.

Yet, in all these works on IA feasibility, it is assumed that every TX knows perfectly the full multi-user channel, which we define as the channel from all the TXs to all the RXs. This assumption is critical as the maximal DoF is known to be significantly lower in the absence of CSIT \cite{Huang2012}. However, a simple examination of IA achievability in particular cases of IC reveals that \emph{how much} CSIT is required at any one TX actually depends on the antenna configuration. Obvious examples include TXs with single antenna which has no alignment capability, hence requires no CSIT, or an IC with many-antenna RXs which eliminates the need for any alignment, hence CSIT. An interesting question is whether interference can be aligned in arbitrary heterogeneous antenna configurations with some TXs having only access to a subset the channel coefficients. To this end, one needs to revisit the IA feasibility question under the prism of CSIT. Note that parallel to this work, the trade-off between the CSI requirements and the DoF has been investigated in \cite{Rao2013}, however in the case of \emph{centralized precoding}. We focus in this work on a completely different problem which is the \emph{incomplete CSIT sharing}, which means that \emph{each TX} receives its \emph{own} CSIT. In particular, we will exploit the fact that some channel coeffients are only known at a subset of TXs, which cannot be considered in \cite{Rao2013}.

To explore this new problem, we introduce a novel CSIT framework whereby CSIT is no longer uniform across TXs. We then say that an IC with $K$~users has \emph{incomplete} CSIT when each of the $K$~TXs acquires, through a given feedback and exchange mechanism left to be specified, a \emph{subset} of the multi-user channel coefficients, with this subset being generally TX-dependent. In this framework, we define the \emph{size} of a CSIT allocation as the total number of scalars forming the CSIT subsets known at the TXs. The main goal of this paper is then to investigate what is the minimal CSIT allocation preserving IA feasibility. We focus here on the single-stream transmission as the feasibility problem of the general case is in itself challenging~\cite{Razaviyayn2012,Ruan2013}.

Specifically, our main contributions are as follows. 
\begin{itemize}
\item We formulate the problem of finding the CSIT allocation of minimal size which preserves IA feasibility. We show conditions under which IA is feasible with strictly incomplete CSIT.
\item For \emph{tightly-feasible} ICs, we present a CSIT allocation policy to the various TXs which preserves IA feasibility while reducing significantly the size of the CSIT feedback.
\item For \emph{super-feasible} ICs, we show the existence of a trade-off between the number of antennas and the CSIT requirements. We provide a heuristic algorithm exploiting any additional antenna to reduce further the size of the CSIT allocation. The code and detailed description of the algorithm are available online \cite{deKerret_online}.
\item As a byproduct of our approach, we develop a new simple and intuitive algorithm for testing the feasibility of single-stream IA. Note that the code for the algorithm is published for convenience in \cite{deKerret_online}.
\end{itemize}
\paragraph* {Notations} We denote the Hadamard (or element-wise product) by the operator~$\odot$ and the Frobenius norm of a matrix~$\bA$ by $\|\bA\|_{\Fro}$. The matrix~$\bm{1}_{n\times m}$ (resp.~$\bm{0}_{n\times m}$) is the matrix made of $n$~rows and $m$~columns with all its elements equal to~$1$ (resp. $0$). We denote the identity matrix of size~$K$ by $\I_K$. We also define the set $\mathcal{K}\triangleq \{1,\ldots,K\}$. The operator $\text{eig}_{\min}\left(\bullet\right)$ returns the eigenvector corresponding to the smallest eigenvalue of the matrix taken as argument. A set containing the elements~$a_1,\ldots,a_n$ is represented by writting its elements inside brackets~$\{a_1,\ldots,a_n\}$. The minus set operator is written $\setminus$ and $|\mathcal{S}|$ is the the number of elements in the set~$\mathcal{S}$. $\CN(0,\sigma^2)$ denotes the zero mean complex circularly symmetric Gaussian distribution of variance~$\sigma^2$. We write \emph{``w.l.o.g."} for \emph{``without loss of generality"} and \emph{``i.i.d."} for \emph{``independent and identically distributed"}.
  
\section{System Model}\label{se:SM}

\subsection{MIMO Interference Channel}\label{se:SM:IC}
We study the transmission in a $K$-user MIMO IC where all the RXs and the TXs are linked by a wireless channel. We consider a conventional channel model with the particularity of our model lying in the structure of the CSIT. We consider that each TX has its \emph{own} CSIT in the form of a sub-matrix of the multi-user channel matrix. In this paper, this specific information structure is referred to as \emph{incomplete CSIT} and will be detailed in Subsection~\ref{se:SM:Incomplete}. TX~$j$ is equipped with $M_j$~antennas, RX~$i$ has $N_i$~antennas, and TX~$j$ transmits a single stream to RX~$j$. This IC is then denoted as $[\prod_{k=1}^K(N_k,M_k)]$. We consider exclusively single-stream transmissions and the extension to multiple streams will be discussed later in this work. When all the TXs and all the RXs have the same (resp. different) number of antennas ,i.e., $[(N,M)^K]$, we say that the antenna configuration is~\emph{homogeneous}  (resp. \emph{heterogeneous}).

The channel from TX~$j$ to RX~$i$ is represented by the channel matrix~$\mathbf{H}_{ij}\in \mathbb{C}^{N_i\times M_j}$ with its elements i.i.d. according to a continuous probability distribution to ensure that all the channel matrices are almost surely full rank. The global multi-user channel matrix is denoted by $\mathbf{H}\in \mathbb{C}^{N_{\Tot}\times M_{\Tot}}$ where $N_{\Tot}\triangleq \sum_{k=1}^K N_k$ and $M_{\Tot}\triangleq \sum_{k=1}^K M_k$:
\begin{equation}
\mathbf{H}\triangleq\begin{bmatrix}
\mathbf{H}_{11}&\mathbf{H}_{12}&\ldots&\mathbf{H}_{1K}\\
\mathbf{H}_{21}&\mathbf{H}_{22}&\ldots&\mathbf{H}_{2K}\\
\vdots&\vdots&\ddots&\vdots\\
\mathbf{H}_{K1}&\mathbf{H}_{K2}&\ldots&\mathbf{H}_{KK}
\end{bmatrix}.
\label{eq:SM_1}
\end{equation}
TX~$i$ uses the unit-norm TX beamformer~$\bm{t}_i\in\mathbb{C}^{M_i\times 1}$ multiplied by~$\sqrt{P}$, where $P$ is the transmit power per-TX, to transmit the data symbol $s_i$ (i.i.d. $\CN(0,1)$) to RX~$i$. The received signal $\bm{y}_i\in \mathbb{C}^{N_i\times 1}$ at the $i$-th RX reads then as
\begin{equation}
\bm{y}_i=\sqrt{P}\mathbf{H}_{ii}\bm{t}_i s_i+\sqrt{P}\sum_{j=1,j\neq i}^K \mathbf{H}_{ij}\bm{t}_j s_j+\bm{\eta}_i
\label{eq:SM_2}
\end{equation}
where $\bm{\eta}_i\in \mathbb{C}^{N_i\times 1}$ is the normalized noise at RX~$i$ and is i.i.d. $\CN(0,1)$. The received signal $\bm{y}_i$ is then processed by a RX filter~$\bm{g}_i^{\He}\in \mathbb{C}^{1\times N_i}$ to obtain an estimate of the data symbol~$s_i$.

Our analysis deals with the achievability of IA which means that the desired signal should be decoded free of interference at each RX. Equivalently, the RX beamformer~$\bm{g}_i^{\He}$ should be able to zero-force (ZF) all the received interference which means fulfilling for all the interferers~$j\neq i$
\begin{equation}
\bm{g}_i^{\He}\mathbf{H}_{ij}\bm{t}_j=0.
\label{eq:SM_3}
\end{equation} 
Thus, IA is feasible if the constraint~\eqref{eq:SM_3} can be achieved at all the RXs for all the interfering streams. Note that this is equivalent to having the interference subspace at RX~$i$ span at most~$N_i-1$ dimensions.

\subsection{Feasibility Results}\label{se:SM:feasibility}

\subsubsection{Results from the literature}

We start by recalling some results from the literature on IA feasibility in a conventional IC with full CSIT sharing for the case of single stream transmission. In \cite{Yetis2010}, the notion of \emph{proper} antenna configurations in introduced. An IC is said to be proper if and only if the number of variables in the RX and TX beamformers involved in any set of IA constraints is larger than the number of scalar equations. Following \cite{Yetis2010}, let us denote by $\mathrm{E}_{ij}$ the IA equation \eqref{eq:SM_3} and by $\Var(\mathrm{E}_{ij})$ the set of free variables involved in this equation. It holds then
\begin{equation}
|\Var(\mathrm{E}_{ij})|=N_i-1+M_j-1.
\label{eq:SM_4}
\end{equation}
A system is said to be proper if and only if
\begin{equation}
|\mathcal{I}|\leq \big|\bigcup_{(i,j)\in \mathcal{I}}\Var(\mathrm{E}_{ij})\big|,\qquad \forall \mathcal{I}\subseteq \mathcal{J}
\label{eq:theorem_feasibility}
\end{equation}
where $\mathcal{J}\triangleq\{(i,j)|1\leq i,j\leq K, i\neq j\}$ and ~$\mathcal{I}$ is an arbitrary subset of~$\mathcal{J}$. In the homogeneous $[(N,M)^K]$ IC, this condition can be reduced to $M+N\geq K+1$. 
The following result has been later obtained in \cite{Razaviyayn2012} and is restated here for convenience.
\begin{theorem}[\cite{Razaviyayn2012}]
\label{theorem_raza}
IA is feasible in the $[\prod_{k=1}^K(N_k,M_k)]$~IC if and only if the antenna configuration is proper, i.e., if and only if equation~\eqref{eq:theorem_feasibility} is verified.
\end{theorem}
Hence, we can use here the condition~\eqref{eq:theorem_feasibility} to determine the feasibility of IA with complete CSIT sharing.

\subsubsection{Tightly-feasible and super-feasible settings}
Whether the total number of variables is strictly larger than the number of equations will be shown to impact significantly the CSIT needed. Hence, we introduce the following definitions.
\begin{definition}
An IC setting is called \emph{tightly-feasible} if this IC is feasible and removing a single antenna at any TX or RX renders IA unfeasible. Equivalently, an IC is tightly-feasible if and only if it is \emph{feasible} and 
\begin{equation}
\sum_{i=1}^K N_i+M_i=K(K+1).
\label{eq:SM_11}
\end{equation} 
\end{definition}
The characterization follows directly from~\eqref{eq:theorem_feasibility} applied with the set~$\mathcal{I}=\mathcal{J}$. 
\begin{definition}
A feasible setting which does not satisfy the tightly-feasible condition is said to be \emph{super-feasible}. Equivalently, a super-feasible setting is a feasible setting such that
\begin{equation}
\sum_{i=1}^K N_i+M_i>K(K+1).
\label{eq:SM_12}
\end{equation}
\end{definition}

\subsubsection{New formulation of the feasibility results}\label{se:SM:feasibility}
Condition \eqref{eq:theorem_feasibility} requires satisfying a number of conditions increasing exponentially with the size of the network. As a preliminary result, we show that condition \eqref{eq:theorem_feasibility} can be simplified to obtain the following condition.

\begin{theorem}
\label{thm_feasibility}
IA is feasible in the $[\prod_{k=1}^K(N_k,M_k)]$~IC if and only if, for any TX subset~$\mathcal{S}_{\TX}$ and any RX subset~$\mathcal{S}_{\RX}$, it holds that
\begin{equation}
\mathcal{N}_{\Var}(\mathcal{S}_{\RX},\mathcal{S}_{\TX})\geq \mathcal{N}_{\Eq}(\mathcal{S}_{\RX},\mathcal{S}_{\TX}),\quad\forall \mathcal{S}_{\TX},\mathcal{S}_{\RX}\subseteq \mathcal{K}
\label{eq:SM_10}
\end{equation}
where~$\mathcal{K}\triangleq \{1,\ldots,K\}$ and $\mathcal{N}_{\Var}(\mathcal{S}_{\RX},\mathcal{S}_{\TX})$ and $\mathcal{N}_{\Eq}(\mathcal{S}_{\RX},\mathcal{S}_{\TX})$ are respectively the number of variables and the number of equations stemming from the subset of RXs~$\mathcal{S}_{\RX}$ and the subset of TXs~$\mathcal{S}_{\TX}$. They are mathematically defined as
\label{thm_feasibility}
\begin{equation}
\begin{aligned}
\mathcal{N}_{\Var}(\mathcal{S}_{\RX},\mathcal{S}_{\TX})&\triangleq \sum_{i\in \mathcal{S}_{\RX}}N_i-1+\sum_{i\in \mathcal	{S}_{\TX}}M_i-1,\\
\mathcal{N}_{\Eq}(\mathcal{S}_{\RX},\mathcal{S}_{\TX})&\triangleq \sum_{k\in \mathcal{S}_{\TX}}\sum_{j\in \mathcal{S}_{\RX},j\neq k}1.
\end{aligned}
\label{eq:SM_8}
\end{equation}   
\end{theorem}
\begin{proof}
A detailed proof is given in Appendix~\ref{app:feasibility}.
\end{proof}
The intuition behind Theorem~$2$ it that it is necessary to test in each IC formed by a subset of TXs and a subset of RXs, which we call a \emph{sub-IC}, that the number of variables is larger than the number of equations. This result is interesting because it is then possible to order the TXs and the RXs to only test condition~\eqref{eq:SM_10} in the sub-ICs with the least number of antennas. 

This leads to a side-contribution of this work which is a simple and intuitive algorithm for testing the feasibility of single-stream IA. Since this algorithm is obtained after very simple modifications of our CSIT allocation algorithm (which will be described later on), it is not given here. A detailed description can be found online in \cite{deKerret_online} along with the MATLAB code.
\begin{remark}
Following Theorem~$1$, it is necessary to test all the subsets~$\mathcal{I}$ included in~$\mathcal{J}$, which means testing a number of subsets increasing exponentially with~$K$ ($2^{K(K-1)}$). However, using the algorithm in \cite{deKerret_online} based on Theorem~$2$ leads to test only a number of subsets which increases polynomially with~$K$.\qed
\end{remark}
\begin{example}
Let us consider as toy-example the $[(1,1).(2,2).(3,3)]$~IC. Following our approach described in \cite{deKerret_online}, it is only necessary to test that there are more variables than equations in the $3$~sub-ICs~$[(1,1)]$, $[(1,1).(2,2)]$, and $[(1,1).(2,2).(3,3)]$. In this example, the ordering is trivial, but the TXs and the RXs can also be ordered in more complicated antenna configurations to obtain similar results.
\end{example}  
An additional interest of Theorem~$2$ is that it provides a useful insight into IA feasibility: The feasibility of IA in the full IC is verified by analyzing the feasibility of IA in all the \emph{sub-ICs} included in the full IC.

Note that the sub-IC obtained after selection of the RXs inside~$\mathcal{S}_{\RX}$ and the TXs inside~$\mathcal{S}_{\TX}$ is not necessarily a conventional IC due to the fact that the TXs and the RXs are not necessarily \emph{paired}. To model this scenario, we introduce in the following the notion of \emph{generalized IC}. 

\subsubsection{Generalized IC}

We refer to an IC in which a TX or a RX does not necessarily have its paired RX or TX included in the IC, as a \emph{generalized IC}. We represent the fact that the node is not part of the IC by writing a ``*" instead of its number of antennas. The IA feasibility criterion \eqref{eq:SM_10} is trivially extended to generalized ICs with the only difference being that only the TXs and the RXs inside the generalized-IC are considered. Note that we will often omit to mention the term ``generalized" when discussing a sub-IC as it is clear that a sub-IC can always be a generalized sub-IC.
\begin{example}
Let us consider the IC~$[(2,2)^3]$. We denote the sub-IC obtained from selecting only RX~$1$, RX~$2$, TX~$1$ and TX~$3$ by~$[(2,2).(2,*).(*,2)]$. Following Theorem~\ref{thm_feasibility}, the feasibility of IA is tested by considering only the sets of RXs included in the set~$\{1,2\}$ and the sets of TXs included in the set~$\{1,3\}$.
\end{example}
\subsection{Incomplete CSIT Model}\label{se:SM:Incomplete} 
The feasibility results from the literature, which we have recalled above, have always made use of an implicit full CSIT assumption. Surprisingly, the problem of revisiting the IA feasibility conditions under the light of a partial CSIT sharing framework has not been addressed before. To fill this gap, it is necessary to introduce a new model to take into account the partial CSIT sharing capability of the TXs.
 
Hence, we consider that a TX has either perfect knowledge of a channel coefficient or no information at all on that element. We represent the CSIT structure at TX~$j$ by the \emph{CSIT matrix}~$\mathbf{A}^{(j)}\in \{0,1\}^{N_{\Tot}\times M_{\Tot}}$ such that~$\{\mathbf{A}^{(j)}\}_{ik}=1$ if $\{\mathbf{H}\}_{ik}$ is known at TX~$j$, and $0$ otherwise. Denoting by~$\mathbf{H}^{(j)}$ the available CSI at TX~$j$, we obtain
\begin{equation}
\mathbf{H}^{(j)}=\mathbf{A}^{(j)} \odot \mathbf{H}.
\label{eq:SM_13}
\end{equation}
We define the CSIT allocation~$\mathcal{A}$ as the set of CSI representations available at all TXs:
\begin{equation}
\mathcal{A}=\{\mathbf{A}^{(j)} | \mathbf{A}^{(j)}\in \{0,1\}^{N_{\Tot}\times M_{\Tot}},j\in \mathcal{K}\}
\label{eq:SM_14}
\end{equation}
and we define the space~$\mathbb{A}$ containing all the possible CSIT allocations. We can then define the \emph{size} of an incomplete CSIT allocation as follows.
\begin{definition}
The size of a CSIT allocation~$\mathcal{A}$, denoted by~$\Size(\mathcal{A})$, is equal to the overall number of complex channel coefficients fed back to the TXs. Thus,
\begin{equation}
\Size(\mathcal{A})\triangleq\sum_{j=1}^K \norm{\mathbf{A}^{(j)}}^2_{\Fro}.
\label{eq:SM_15}
\end{equation}
\end{definition}  

To check whether the IA feasibility is preserved with a given CSIT allocation, we introduce the function~$f_{\Feas}$ which takes as argument a CSIT allocation~$\mathcal{A}$ and an antenna configuration~$[\prod_{k=1}^K(N_k,M_k)]$ and returns~$1$ if IA is feasible with these parameters and~$0$ otherwise. Note that this means that there exists one algorithm achieving IA with this CSIT allocation but it does not precise the algorithm. We also define the set~$\mathbb{A}_{\Feas}$ containing all the CSIT allocations for which IA is feasible. Hence,
\begin{equation}
\mathbb{A}_{\Feas}\triangleq\bigg\{\mathcal{A}|\mathcal{A}\in \mathbb{A},f_{\Feas}\bigg(\mathcal{A},\bigg[\prod_{k=1}^K(N_k,M_k)\bigg]\bigg)=1\bigg\}.
\label{eq:SM_16}
\end{equation}
Note that only the interfering channel matrices~$\mathbf{H}_{ij}$ with~$i\neq j$ are required to fulfill the IA constraints, and not the direct channel matrices~$\mathbf{H}_{jj}$. Thus, from a DoF point of view, we can always skip the direct channel matrices~$\mathbf{H}_{jj}$ in the feedback, which leads to the following definition.
\begin{definition}
A \emph{complete} CSIT allocation, denoted by~$\mathcal{A}_{\Comp}$, is defined by the knowledge of all the interfering channel matrices~$\mathbf{H}_{ij}$ with~$i\neq j$ at all TXs. Thus, the size of a complete CSIT allocation is
\begin{equation}
\Size\left(\mathcal{A}_{\Comp}\right)=K\left(N_{\Tot} M_{\Tot}-\sum_{i=1}^{K}N_i M_i\right).
\label{eq:SM_17}
\end{equation}
A CSIT allocation with a size smaller than~$\Size(\mathcal{A}_{\Comp})$ is said to be \emph{strictly incomplete}.
\end{definition}
At this stage, a natural question is to ask what is the most incomplete CSIT allocation which preserves the feasibility of IA, i.e., to find
\begin{equation}
\mathcal{A}_{\min}=\argmin_{\mathcal{A}\in \mathbb{A}_{\Feas}} \Size(\mathcal{A}).
\label{eq:SM_18}
\end{equation}
Note that we limit here our study to the IA feasible settings, i.e., such that~$\mathcal{A}_{\Comp}\in\mathbb{A}_{\Feas}$.

\section{IA with Incomplete CSIT for Tightly-Feasible Channels}\label{se:Tight} 

\subsection{General Criterion}\label{se:Tight:Feasibility}
\subsubsection{Parametrization of the CSIT allocation}\label{se:Tight:Minimal}

The incomplete CSIT model described in Subsection~\ref{se:SM:Incomplete} allows for any TX to receive the feedback of any channel coefficient. However, we will show in this work that it is not meaningful (with the precoding algorithms considered here) to feedback to a given TX only some coefficients of the matrix~$\bH_{ij}$. Hence, only the CSIT allocations which can be written under the following form will be of interest to us in this work.

We define the matrix~$\mathbf{A}_{\mathcal{S}_{\RX},\mathcal{S}_{\TX}}$, where~$\mathcal{S}_{\RX}$ is a set of RXs and~$\mathcal{S}_{\TX}$ a set of TXs, such that $\mathbf{A}_{\mathcal{S}_{\RX},\mathcal{S}_{\TX}} \odot \mathbf{H}$ contains all the channel coefficients relative to the generalized sub-IC formed by the set of RXs~$\mathcal{S}_{\RX}$ and the set of TXs~$\mathcal{S}_{\TX}$, at the exception of the direct channel matrices~$\mathbf{H}_{jj},\forall j$.
Mathematically, this means that the matrix~$\mathbf{A}_{\mathcal{S}_{\RX},\mathcal{S}_{\TX}}$ of size~$N_{\Tot}\times M_{\Tot}$ has its only nonzero elements chosen to satisfy~$\forall x\neq y,x\in \mathcal{S}_{\RX}, y\in \mathcal{S}_{\TX},$
\begin{equation} 
\left(\mathbf{E}_{\RX}^x\right)^{\trans}\mathbf{A}_{\mathcal{S}_{\RX},\mathcal{S}_{\TX}}\mathbf{E}_{\TX}^y=\left(\mathbf{E}_{\RX}^{x}\right)^{\trans}\mathbf{1}_{N_{\Tot}\times M_{\Tot}}\mathbf{E}_{\TX}^	{y},
\label{eq:Tight_1}
\end{equation}
with $\mathbf{E}_{\TX}^n\triangleq \begin{bmatrix}\mathbf{0}_{\sum_{k=1}^{n-1}M_k\times M_n},\I_{M_n},\mathbf{0}_{\sum_{k=n+1}^{K}M_k\times M_n}\end{bmatrix}^{\trans}$ and the matrix $\mathbf{E}_{\RX}^n$ defined similarly with $N_i$ replacing $M_i$. Note that if either $\mathcal{S}_{\RX}$ or $\mathcal{S}_{\TX}$ is empty, the matrix~$\mathbf{A}_{\mathcal{S}_{\RX},\mathcal{S}_{\TX}}$ contains only $0$, i.e.,~$\mathbf{A}_{\mathcal{S}_{\RX},\mathcal{S}_{\TX}}=\bm{0}_{N_{\Tot}\times M_{\Tot}}$. 

\begin{remark}
Using this parameterization restricts the possible CSIT allocations. However, we will show that considering these CSIT allocations is sufficient to achieve significant gains. In fact, it is believed that there is no loss incurred by this parameterization.\qed
\end{remark}

\subsubsection{Main theorem}\label{se:Tight:Feasibility}
We can now state one of our main results.
\begin{theorem}
\label{thm_incomplete}
In a tightly-feasible $[\prod_{k=1}^K (N_k,M_k)]$~IC, if there exists a tightly-feasible sub-IC formed by the set of TXs~$\mathcal{S}_{\TX}$ and the set of RXs~$\mathcal{S}_{\RX}$, i.e.,
\begin{equation}
\mathcal{N}_{\Var}(\mathcal{S}_{\RX},\mathcal{S}_{\TX})=\mathcal{N}_{\Eq}(\mathcal{S}_{\RX},\mathcal{S}_{\TX}),
\label{eq:Tight_3}
\end{equation} 
then the incomplete CSIT allocation~$\mathcal{A}=\{\mathbf{A}^{(j)}|j\in \mathcal{K}\}$ preserves IA feasibility, i.e., $\mathcal{A}\in \mathbb{A}_{\Feas}$, if
\begin{equation}
\begin{aligned}
\mathbf{A}^{(j)}&=\mathbf{A}_{\mathcal{S}_{\RX},\mathcal{S}_{\TX}},\qquad&\forall j\in \mathcal{S}_{\TX}\\
\mathbf{A}^{(j)}&=\mathbf{A}_{\mathcal{K},\mathcal{K}}=\mathbf{1}_{N_{\Tot}\times M_{\Tot}},\qquad&\forall j\notin \mathcal{S}_{\TX}.
\end{aligned}
\label{eq:Tight_4}
\end{equation} 
\end{theorem}
\begin{proof}
A detailed proof is provided in Appendix~\ref{app:thm_incomplete}.
\end{proof}
This theorem implies that if there exists a tightly-feasible sub-IC \emph{strictly} included in the considered IC, then there also exists a \emph{strictly} incomplete CSIT allocation preserving IA feasibility.  
\begin{example}
Let us consider as toy-example the $[(2,2).(2,2).(2,2).(4,4)]$~IC. We can easily observe that the first $3$ TX/RX pairs form the well-known tightly-feasible $[(2,2)^3]$~IC. Hence, if the $3$ first TXs align interference inside this sub-IC, RX~$4$ has then enough antenna to remove all its received interference. In addition, TX~$4$ can use its $4$~antennas to eliminate the interference that it emits to the signal subspaces at the first $3$~RXs.
\end{example}


In fact, it can be easily seen that the obtained incomplete CSIT allocation exploits the \emph{heterogeneity} of the antenna configuration. Indeed, there can be a tightly-feasible sub-IC \emph{strictly} included in a tightly-feasible IC only if the antenna configuration is heterogeneous.

\begin{corollary}
\label{corollary_homogeneous}
In the homogeneous tightly-feasible $[(N,M)^K]$~IC (this implies $M+N=K+1$) with~$M\neq1$ and~$M\neq K$, there exists no generalized tightly-feasible sub-IC strictly included in the IC. Hence, the previous sufficient condition leads to no CSIT reduction.
\end{corollary}
\begin{proof} 
The proof follows easily by evaluating \eqref{eq:Tight_3} in an homogeneous setting and is omitted for brevity.
\end{proof}
This property only holds for tightly-feasible settings and we will show in the following section that CSIT reductions can be achieved for super-feasible ICs in any antenna configuration. It is then not the antenna heterogeneity which is exploited, but the additional antennas.

\subsection{Example of tightly-feasible configuration}\label{se:Tight:Ex}

Applying iteratively Theorem~\ref{thm_incomplete} leads to a CSIT allocation algorithm which we will describe in Subsection~\ref{se:Tight:Algo}. The corresponding problem of designing an algorithm achieving IA based on the incomplete CSIT allocation will then be tackled in Subsection~\ref{se:Tight:IA_precoding}. But before providing the algorithms, we describe now in a small example how our approach works, so as to gain insight into the problem.

Let us consider the IC formed by the antenna configuration~$[(2,3).(2,4).(3,5).(3,2).(4,2)]$. The CSIT allocation algorithm (which will be presented in Subsection~\ref{se:Tight:Algo}) returns 
\begin{equation*} 
\begin{aligned}
\mathcal{A}&=\bigg\{\mathbf{A}^{(1)}=\mathbf{A}_{\{1,2,3\},\{4,5,1\}},\mathbf{A}^{(2)}=\mathbf{A}_{\{1,2,3,4\},\{1,2,4,5\}}\\
&,\mathbf{A}^{(3)}\!=\!\mathbf{A}_{\mathcal{K},\mathcal{K}},\mathbf{A}^{(4)}\!=\!\mathbf{A}_{\{1,2\},\{4,5\}},\mathbf{A}^{(5)}\!=\!\mathbf{A}_{\{1,2\},\{4,5\}}\bigg\}
\end{aligned}
\label{eq:ExTight_2} 
\end{equation*}
We remind the reader that the notation~$\mathbf{A}^{(4)}=\mathbf{A}_{\{1,2\},\{4,5\}}$ means that TX~$4$ receives the CSI relative to the sub-IC formed by the TXs in the set $\{4,5\}$ and the RXs in the set~$\{1,2\}$.

Hence, TX~$4$ and TX~$5$ have only the CSI sufficient to align their interference at RX~$1$ and RX~$2$, which is in fact the first step of the IA algorithm. Once this is done, TX~$1$ designs its beamformer to align its interference on the interference subspace created by TX~$4$ and TX~$5$ at RX~$2$ and RX~$3$. Note that it has a sufficient CSI to do so. Proceeding further, TX~$2$ aligns its interference on the interference subspace spanned at RX~$1$, RX~$3$, and RX~$4$ by the previous TX beamformers. At this step, all the interference subspaces have been generated and TX~$3$ uses its~$5$ antennas to align its interference at all the RXs. 

The general idea is very simple and reads as follows: The TX beamformers in the smallest tightly-feasible ICs are computed first until all the TX beamformers are computed. Note that the size of the incomplete CSIT allocation obtained in the previous example is equal to~$346$ while the complete CSIT allocation has a size of~$905$.

\subsection{CSIT Allocation Algorithm }\label{se:Tight:Algo}
The CSIT allocation algorithm takes as input the antenna configuration~$[\prod_{k=1}^K (N_k,M_k)]$ and returns as output the incomplete CSIT allocation~$\mathcal{A}=\{\mathbf{A}^{(j)}|j\in \mathcal{K}\}$ such that	
\begin{equation}
\mathbf{A}^{(j)}=\mathbf{A}_{\mathcal{S}_{\RX}^{(j)},\mathcal{S}_{\TX}^{(j)}}, \forall j 
\end{equation}
with $\mathcal{S}_{\RX}^{(j)},\mathcal{S}_{\TX}^{(j)}\subset \mathcal{K}$.
\begin{remark}
With simple words, this algorithm finds all the tightly-feasible sub-ICs and allocates to each TX the CSI relative to the smallest tightly-feasible sub-IC to which it belongs.\qed
\end{remark}
Let us consider w.l.o.g. the problem of allocating the CSI to TX~$j$.

\textbf{Initialization: }
We first define an initial pair of sets~$\mathcal{S}_{}\triangleq (\mathcal{S}_{\RX},\mathcal{S}_{\TX})$ initialized such that
\begin{equation}
\mathcal{S}_{}=(\emptyset,\{j\}).
\label{eq:algo_1}
\end{equation}
The remaining TXs (without considering TX $j$) are ordered by increasing number of antennas, i.e., with the permutation $\sigma_{\TX}$ satisfying
\begin{equation}
M_{\sigma_{\TX}(i)}\leq M_{\sigma_{\TX}(i+1)},\qquad \forall i\in\{1,\ldots,K-2\}
\label{eq:algo_2}
\end{equation} 
and symmetrically, the RXs are ordered by increasing number of antennas, i.e., with the permutation $\sigma_{\RX}$ satisfying
\begin{equation}
N_{\sigma_{\RX}(i)}\leq N_{\sigma_{\RX}(i+1)},\qquad \forall i.
\label{eq:algo_3}
\end{equation}
In case of equality, we order the TXs to ensure that  
\begin{equation}
(M_{\sigma_{\TX}(i)}= M_{\sigma_{\TX}(i\!+\!1)})\Rightarrow N_{\sigma_{\TX}(i)}\geq N_{\sigma_{\TX}(i\!+\!1)},\quad \forall i.
\label{eq:algo_4}
\end{equation}
Similarly, the RX ordering is modified to ensure that
\begin{equation}
(N_{\sigma_{\RX}(i)}= N_{\sigma_{\RX}(i\!+\!1)})\Rightarrow M_{\sigma_{\RX}(i)}\geq M_{\sigma_{\RX}(i\!+\!1)},\quad \forall i.
\label{eq:algo_5}
\end{equation} 
In case both the two TXs and their matched RXs have the same number of antennas, the RX ordering~$\sigma_{\RX}$ is modified to ensure that the RXs are ordered in the opposite of the TXs, i.e.,
\begin{figure*}
\begin{equation}
\LB M_{\sigma_{\TX}(i)}= M_{\sigma_{\TX}(i+1)},N_{\sigma_{\TX}(i)}= N_{\sigma_{\TX}(i+1)}\RB \Rightarrow \LB\sigma_{\RX}^{-1}(\sigma_{\TX}(i+1))<\sigma_{\RX}^{-1}(\sigma_{\TX}(i))\RB,\quad \forall i.
\label{eq:algo_4}
\end{equation} 
\end{figure*}
\begin{remark}
These two permutations have been defined such that selecting the TXs and the RXs respectively according to $\sigma_{\TX}$ and $\sigma_{\RX}$ will lead to select the TXs and the RXs with the smallest number of antennas with non-matched TXs and RXs in case of equality in the number of antennas. This can easily be seen to ensure that the ``most tight" sub-ICs are selected.\qed
\end{remark}

\textbf{Update at step~$n$: }
 Let us assume that we are given the pair of sets~$\mathcal{S}_{}=(\mathcal{S}_{\RX},\mathcal{S}_{\TX})$.
\begin{enumerate}
\item If $\mathcal{N}_{\Var}(\mathcal{S}_{\RX},\mathcal{S}_{\TX})=\mathcal{N}_{\Eq}(\mathcal{S}_{\RX},\mathcal{S}_{\TX})$ is satisfied by the sets~$\mathcal{S}_{\RX}$ and~$\mathcal{S}_{\TX}$, the sub-IC obtained is tightly-feasible and the algorithm has reached its end. We set $\mathcal{S}_{\RX}^{(j)}=\mathcal{S}_{\RX}$, $\mathcal{S}_{\TX}^{(j)}=\mathcal{S}_{\TX}$ and
\begin{equation} 
\mathbf{A}^{(j)}=\mathbf{A}_{\mathcal{S}_{\RX}^{(j)},\mathcal{S}_{\TX}^{(j)}}. 
\label{eq:algo_6}
\end{equation} 
\item If $\mathcal{N}_{\Var}(\mathcal{S}_{\RX},\mathcal{S}_{\TX})\neq\mathcal{N}_{\Eq}(\mathcal{S}_{\RX},\mathcal{S}_{\TX})$, we verify whether adding the next RX adds more equations than variables, i.e., whether
\begin{equation} 
\begin{aligned}
&\mathcal{N}_{\Var}(\mathcal{S}_{\RX},\mathcal{S}_{\TX})- \mathcal{N}_{\Eq}(\mathcal{S}_{\RX},\mathcal{S}_{\TX})\\&\geq  \mathcal{N}_{\Var}(\{\mathcal{S}_{\RX},\sigma_{\RX}(|\mathcal{S}_{\RX}|+1)\},\mathcal{S}_{\TX})\\
&- \mathcal{N}_{\Eq}(\{\mathcal{S}_{\RX},\sigma_{\RX}(|\mathcal{S}_{\RX}|+1)\},\mathcal{S}_{\TX})
\end{aligned}
\label{eq:algo_7}
\end{equation}
\begin{itemize}
\item If~\eqref{eq:algo_7} is satisfied, we set
\begin{equation} 
\mathcal{S}_{\RX}^{}=\{\mathcal{S}_{\RX},\sigma_{\RX}(|\mathcal{S}_{\RX}|+1)\}
\label{eq:algo_8}
\end{equation}
and we start over at step~$n+1$.
\item If \eqref{eq:algo_7} is not satisfied, then
\begin{itemize}
\item If~$|\mathcal{S}_{\TX}|<K$, we increase the set of TXs as 
\begin{equation} 
\mathcal{S}_{\TX}=\left\{\mathcal{S}_{\TX},\sigma_{\TX}(|\mathcal{S}_{\TX}|+1)\right\}
\label{eq:algo_9}
\end{equation}
and we start over at step~$n+1$.
\item If~$|\mathcal{S}_{\TX}|=K$, then the algorithm has reached its end and we set $\mathcal{S}_{\RX}^{(j)}=\mathcal{S}_{\RX}$ and $\mathcal{S}_{\TX}^{(j)}=\mathcal{S}_{\TX}$ and
\begin{equation} 
\mathbf{A}^{(j)}=\mathbf{A}_{\mathcal{S}_{\RX}^{(j)},\mathcal{S}_{\TX}^{(j)}}. 
\label{eq:algo_10}
\end{equation}
\end{itemize}
\end{itemize} 
\end{enumerate}
 
\subsection{IA Algorithm for Incomplete CSIT Allocation} \label{se:Tight:IA_precoding}
We consider now the CSIT allocation~$\mathcal{A}$ to be given and we describe a novel IA algorithm which achieves IA using an adequate incomplete CSIT allocation. The description of the algorithm is split into two parts: In Sub-subsection~\ref{se:Tight:IA_precoding:effective}, an algorithm forming a building block of the total precoding function is described, while it is shown in Sub-subsection~\ref{se:Tight:IA_precoding:incomplete} how this sub-algorithm is used to design the IA precoder. Finally, it is demonstrated in Sub-subsection~\ref{se:Tight:Min} that the proposed algorithm achieves IA. 

The IA precoding algorithm runs in a distributed fashion at each TX and is denoted by~$f_{\IA}$. It takes as input the antenna configuration, the CSIT allocation policy, and the channel coefficients known at the TX, and returns the beamformer for this TX. Thus, we can write at TX~$j$
\begin{equation}
\bm{t}_j=f_{\IA}\bigg(\bigg[\prod_{k=1}^K(N_k,M_k)\bigg],\mathcal{A},\mathbf{H}^{(j)}\bigg).
\label{eq:algo_11}
\end{equation}

\subsubsection{IA algorithm for the effective channel}\label{se:Tight:IA_precoding:effective}
We start by introducing an IA algorithm~$f_{\Eff}$ which will be a building block for our algorithm. It consists in running an IA algorithm over the \emph{effective channel}, which we define as the channel obtained once a fraction of the TX beamformers have been fixed. 
\begin{example}
Let us consider the $[(2,2).(2,2).(2,2)]$~IC and that the TX beamformer of TX~$1$ has been fixed. The resulting \emph{effective channel} is equal to the initial channel with the difference that the channel from TX~$1$ to RX~$i$ is given by $\bH_{i1}\bt_1,\forall i$. The antenna configuration of the effective channel is then~$[(2,1).(2,2).(2,2)]$.
\end{example}
Taking as input the set containing the fixed beamformers~$\mathcal{B}_{\TX}^{\Fix}$ and a channel matrix~$\mathbf{G}$, it returns as output the set of beamformers $\mathcal{B}_{\TX}$ obtained after having run a conventional IA algorithm from the literature over this effective channel. Note that since the TX beamformers inside~$\mathcal{B}_{\TX}^{\Fix}$ are not modified, it holds that~$\mathcal{B}_{\TX}^{\Fix}\subset \mathcal{B}_{\TX}$. We can then write
\begin{equation}
\mathcal{B}_{\TX}=f_{\Eff}\big(\mathbf{G},\mathcal{B}_{\TX}^{\Fix}\big).
\label{eq:algo_12}
\end{equation} 
A number of IA algorithms can be run over the effective channel, and we will use the most simple IA algorithm called the \emph{min-leakage} algorithm\cite{Gomadam2008}. We recall for completeness its main steps in Appendix~\ref{app:minLeak}. Our IA algorithm is obtained from the min-leakage algorithm after two simple modifications of the update formulas [Cf. equations \eqref{eq:ML_2} and \eqref{eq:ML_3}]:
\begin{itemize}
\item The update of the beamformers on the RX side (resp. on the TX side) is done by summing over all the interfering TXs (resp. RXs) and not from~$1$ to $K$ because there are not necessarily $K$~TXs or $K$~RXs.
\item The TX beamformers contained in $\mathcal{B}_{\TX}^{\Fix}$ are kept unchanged.
\end{itemize} 

\subsubsection{Precoding with incomplete CSIT} \label{se:Tight:IA_precoding:incomplete}

Let us consider now the precoding at TX~$j$ with the CSIT allocation~$\mathbf{H}^{(j)}=\mathbf{A}_{\mathcal{S}_{\RX}^{(j)},\mathcal{S}_{\TX}^{(j)}}\odot \mathbf{H}$. We define now in a recursive manner the precoding algorithm $f_{\IA}$ introduced in \eqref{eq:algo_11}.

We start by defining the set~$\mathcal{C}_j$ containing all the TXs whose CSIT allocations are strictly included in the CSIT known at TX~$j$. Hence the set~$\mathcal{C}_j$ is defined as
\begin{equation}
\mathcal{C}_j\triangleq \{k|\mathcal{S}_{\RX}^{(k)}\subsetneq \mathcal{S}_{\RX}^{(j)},\mathcal{S}_{\TX}^{(k)}\subsetneq \mathcal{S}_{\TX}^{(j)}\}.
\label{eq:algo_13}
\end{equation}
The beamformer~$\bm{t}_j$ is then obtained from
\begin{equation}
\bm{t}_j=f_{\Eff}\big(\tilde{\mathbf{H}}^{(j)},\{\bm{t}_k\}_{k\in \mathcal{C}_j}\big)
\label{eq:algo_14}
\end{equation}
where~$\tilde{\mathbf{H}}^{(j)}$ is the submatrix of~${\mathbf{H}}^{(j)}$ containing only the columns and rows which are nonzero, and the beamformers $\{\bm{t}_k\}_{k\in \mathcal{C}_j}$ are obtained from
\begin{equation}
\bm{t}_k=f_{\IA}\bigg(\bigg[\prod_{k=1}^K(N_k,M_k)\bigg],\mathcal{A},\mathbf{H}^{(k)}\bigg),\qquad \forall k \in \mathcal{C}_j.
\label{eq:algo_15}
\end{equation}
Note that if $\mathcal{C}_j= \emptyset$, the beamformer~$\bm{t}_j$ is simply obtained from~$\bm{t}_j=f_{\Eff}(\tilde{\mathbf{H}}^{(j)},\emptyset)$.
 
\begin{remark}
TX~$j$ computes first the TX beamformers of all the TXs which have a CSIT included in its own CSIT allocation. They belong to smaller tightly-feasible ICs and TX~$j$ has to align its interference over the interference subspace that they generate.
\qed
\end{remark}
\subsubsection{Achievability of interference alignment}\label{se:Tight:Min}
We have described a precoding algorithm but it remains to prove that IA is indeed achieved.
\begin{theorem}
\label{thm_Min}
The CSIT allocation policy~$\mathcal{A}$ obtained with the incomplete CSIT allocation algorithm described above preserves IA feasibility, i.e., it holds that~$\mathcal{A}\in \mathbb{A}_{\Feas}$.
\end{theorem}
\begin{proof}  
A detailed proof is provided in Appendix~\ref{app:thm_Min}.
\end{proof}
\begin{remark} The RXs in~$\mathcal{S}_{\RX}^{(j)}$ and the TXs in~$\mathcal{S}_{\TX}^{(j)}$, as returned by the CSIT allocation algorithm, form together the smallest tightly-feasible setting containing TX~$j$. If the algorithm is modified such that the initialization is~$\mathcal{S}_{\TX}^{(j)}=\emptyset$ instead of~$\mathcal{S}_{\TX}^{(j)}=\{j\}$, the smallest tightly-feasible sub-IC is obtained and if IA is not feasible, a sub-IC where IA is not feasible is found. Hence, this algorithm can also be used to verify the IA feasibility of an antenna configuration, as described in \cite{deKerret_online}.\qed
\end{remark}

\section{Interference Alignment with Incomplete CSIT for Super-Feasible Channels}\label{se:NTight}

The previous section indicates how CSIT savings can be obtained for tightly-feasible scenarios. When additional antennas are available, the intuition goes that further CSIT savings should be possible at no cost in terms of IA feasibility. We now investigate this question. 

A distinct feature of super-feasible settings is that there must exist a corresponding tightly-feasible setting that can be obtained by keeping all TXs and RXs and simply ignoring certain antennas among the overall antenna set. Clearly, there are generally multiple ways for arriving at a tightly-feasible setting from a super-feasible one. Depending on the choice of which antennas are ignored in the initial super-feasible setting, the obtained tightly-feasible settings will satisfy particular CSIT requirements. 

As a consequence, instead of considering directly optimization problem~\eqref{eq:SM_18}, we consider the following optimization problem :
\begin{equation}
\begin{aligned}
\mathcal{A}=\argmin_{\mathcal{A}\in \mathbb{A}} & \min_{\prod_{k=1}^K(N_k',M_k')}\Size(\mathcal{A}) \quad \\
&\text{s.t. $f_{\Feas}\bigg(\mathcal{A},\bigg[\prod_{k=1}^K(N_k',M_k')\bigg]\bigg) =1$} \\
&\text{s.t. $\sum_{i=1}^K  M_i'+N_i'=(K+1)K	$}\\
&\text{s.t. $1\leq M_i'\leq M_i$ and $1\leq N_i'\leq N_i$}.
\end{aligned}
\label{eq:NTight_1} 
\end{equation}  
The problem of finding the minimal CSIT allocation has been reduced to finding the tightly-feasible setting (containing all the users) included in the full super-feasible setting, which requires the smallest CSIT allocation. Since a CSIT allocation algorithm has been derived for tightly-feasible settings, it remains only to determine which RXs or TXs should not fully exploit their antennas to ZF interference dimensions, i.e., where some antennas should be ``removed" in terms of IA feasibility.

\begin{remark}
Practically, the antennas are not removed but some precoding dimensions are used for another purpose than aligning interference inside the IC (e.g., reducing interference to other RXs, increasing signal power, diversity, etc...). As an example, we will now show how it can be used to increase the received signal power. Intuitively, we select the precoding subspace of dimension $n$ with $n<M_i$ which provides the largest received power to the RX. As a consequence, the quality of the direct channel is improved. Let us write the singular value decomposition of $\mathbf{H}_{ii}\in \mathbb{C}^{N_i\times M_i}$ as $\mathbf{H}_{ii}=\mathbf{U}_i\bm{\Sigma}_i\mathbf{V}_i^{\He}$ with $\mathbf{V}_i=[\bm{v}_1,\ldots,\bm{v}_{M_i}]\in \mathbb{C}^{M_i\times M_i}$ and $\mathbf{U}_i=[\bm{u}_1,\ldots,\bm{u}_{N_i}]\in \mathbb{C}^{N_i\times N_i}$ being two unitary matrices and $\bm{\Sigma}_i=\diag(\sigma_1,\ldots,\sigma_{\min(M_i, N_i)},0,\ldots,0)$. We set $\bm{t}_i=[\bm{v}_1,\ldots,\bm{v}_{n}]\bm{t}'_i$ with $\bm{t}'_i\in \mathbb{C}^{n\times 1}$ such that the dimension of the precoding subspace is reduced from $M_i$ to $n$. However, the vectors~$\bm{v}_1,\ldots,\bm{v}_n$ span the subspace of dimension $n$ with the largest power. Altogether, the number of dimensions available for ZF precoding is reduced by one, which is equivalent in terms of IA feasibility to removing one antenna, while the quality of the direct channel is improved. \footnote{Note that this step can be applied similarly on the RX side and that this process on the TX side requires the CSI relative to the direct channel.}\qed
\end{remark}
The considered optimization problem is combinatorial in the total number of TXs and RXs which makes exhaustive search only practical for small settings. As a consequence, we provide in the following a CSIT allocation policy exploiting heuristically the additional antennas available to reduce the size of the CSIT allocation. The heuristic behind the algorithm comes from the insight gained in the analysis of tightly-feasible settings that the more heterogeneous the antenna configuration is, the smaller the size of the CSIT allocation becomes. Intuitively, our algorithm ``removes" the antennas so as to form the ``most heterogeneous" antenna configuration where IA remains feasible.

\subsection{CSIT Allocation Algorithm}\label{se:CSI:algo} 

We consider in the following an heterogeneous IC and we denote by~$S$ the total number of additional antennas in the sense that~$S$ is defined as
\begin{equation}
S\triangleq\sum_{i=1}^K M_i+N_i-(K+1)K.
\label{eq:NTight_9}
\end{equation}

The following algorithm will provide the pair of sets $\mathcal{S}^{\NT}=(\mathcal{S}^{\NT}_{\RX},\mathcal{S}^{\NT}_{\TX})$ containing respectively the RXs and the TXs where the additional antennas should be ``removed". Once these antennas have been removed, the incomplete CSIT allocation policy for tightly-feasible settings described in Section~\ref{se:Tight} can be applied to obtain the incomplete CSIT allocation. Note that we need to ensure that IA feasibility is preserved by the removing of the antennas. The algorithm relies on the same approach as the algorithm for tightly-feasible setting. For the sake of brevity, we will therefore only present briefly the main steps of the algorithm. The detailed description can be found online in \cite{deKerret_online} along with the MATLAB code.

\textbf{Initialization: } We start with the initialization~$\mathcal{S}^{\NT}=(\mathcal{S}^{\NT}_{\RX},\mathcal{S}^{\NT}_{\TX})=\{\emptyset,\emptyset\}$. 

\textbf{Step~$n$: } We introduce the antenna configuration with the antennas already ``removed" as $[\prod_{i=1}^K(N_i',M_i')]$. It is initialized equal to $[\prod_{i=1}^K(N_i,M_i)]$ and updated such that
\begin{equation}
\begin{aligned}
N'_i&=N'_i-1,&\qquad\forall i\in \mathcal{S}_{\RX}^{\NT},\\
M'_i&=M_i'-1,&\qquad\forall i\in \mathcal{S}_{\TX}^{\NT}.
\end{aligned}
\end{equation}
\begin{enumerate}
\item In~$[\prod_{i=1}^K(N_i',M_i')]$, we find the set of TXs and the set of RXs, denoted by~$(\mathcal{S}^{\mathrm{Tight}}_{\RX}(n),\mathcal{S}^{\mathrm{Tight}}_{\TX}(n))$, containing the TXs and the RXs which belong to at least one tightly-feasible sub-IC.
\item We add to the set $\mathcal{S}^{\NT}_{\TX}$ the TX with the smallest number of antenna in the set ~$\mathcal{K}\setminus \mathcal{S}^{\mathrm{Tight}}_{\TX}(n)$. If the set~$\mathcal{K}\setminus \mathcal{S}^{\mathrm{Tight}}_{\TX}(n)$ is empty, we instead add to the set $\mathcal{S}^{\NT}_{\RX}$ the RX with the smallest number of antenna in the set~$\mathcal{K}\setminus \mathcal{S}^{\mathrm{Tight}}_{\RX}(n)$.
\end{enumerate}

\textbf{Discussion:} Procedure~$1)$ selects all the nodes where it is \emph{not} possible to remove one antenna without making IA unfeasible. It is very similar to the algorithm described in Section~\ref{se:Tight}. Procedure $2)$ decides at which node the antenna should be removed. We have proposed here one heuristic policy but other heuristic policies could as well be chosen.
 
\subsection{Toy-Example of the Incomplete CSIT-Algorithm in Super-Feasible Settings}\label{se:CSI:algo} 
Let us consider the $[(2,2).(3,2).(2,3)]$~IC. It can be easily verified to be a feasible IA setting. Furthermore, it contains two additional antennas since $\sum_{i=1}^KN_i+M_i-K(K+1)=2$. We will now go through the steps of our CSIT allocation algorithm for super-feasible ICs. 
\begin{itemize}
\item $n=1$: During phase $1)$, it is found that there is not any tightly-feasible set. Thus, one antenna can be removed at any node during phase $2)$. Hence, one antenna is removed at TX~$1$. 
\item $n=2$: The $[\prod_{i=1}^K(N_i',M_i')]$~IC is then equal to $[(2,1).(3,2).(2,3)]$. During phase~$1)$, the set of TXs belonging to a tightly-feasible IC is found to be~$\mathcal{S}^{\mathrm{Tight}}_{\TX}(n)=\{1,2\}$ and the set of RXs to be~$\mathcal{S}^{\mathrm{Tight}}_{\RX}(n)=\{1,3\}$. Hence, one antenna can be removed during phase $2)$ at TX~$3$ while preserving IA feasibility.
\end{itemize}

The CSIT allocation algorithm leads to remove one antenna at TX~$1$ and one antenna at TX~$3$ to obtain the antenna configuration~$[(2,1).(3,2).(2,2)]$. This setting being tightly-feasible, we can run the CSIT allocation for tightly-feasible ICs described in Subsection~\ref{se:Tight:Algo} which returns the CSIT allocation
\begin{equation*} 
\mathcal{A}=\{\mathbf{A}^{(1)}\!=\!\mathbf{A}_{\emptyset,\emptyset},\mathbf{A}^{(2)}\!=\!\mathbf{A}_{\{3\},\{1,2\}},\mathbf{A}^{(3)}\!=\!\mathbf{A}_{\{1,3\},\{1,2,3\}}\}.
\label{eq:Ex_NT_1}
\end{equation*}
The size of the CSIT allocation in~\eqref{eq:Ex_NT_1} is equal to~$20$ while the complete CSIT allocation has a size of~$99$. Thus, the additional antennas have been used to reduce the feedback size by practically a factor of $4$.

\section{Simulations}

\subsection{Tightly-Feasible Setting}
We start by verifying by simulations that IA is indeed achieved by our new IA algorithm. We consider for the simulations the $[(2,3).(2,4).(3,5).(3,2).(4,2)]$~IC, which has been studied in the example in Subsection~\ref{se:Tight:Ex}. This example has been chosen to illustrate our approach, but the CSIT reduction is different for each antenna configuration such that it is also relevant to consider the average reduction over all the antenna configurations. This will be the focus of the next subsection.
\begin{figure}
\centering
\includegraphics[width=1\columnwidth]{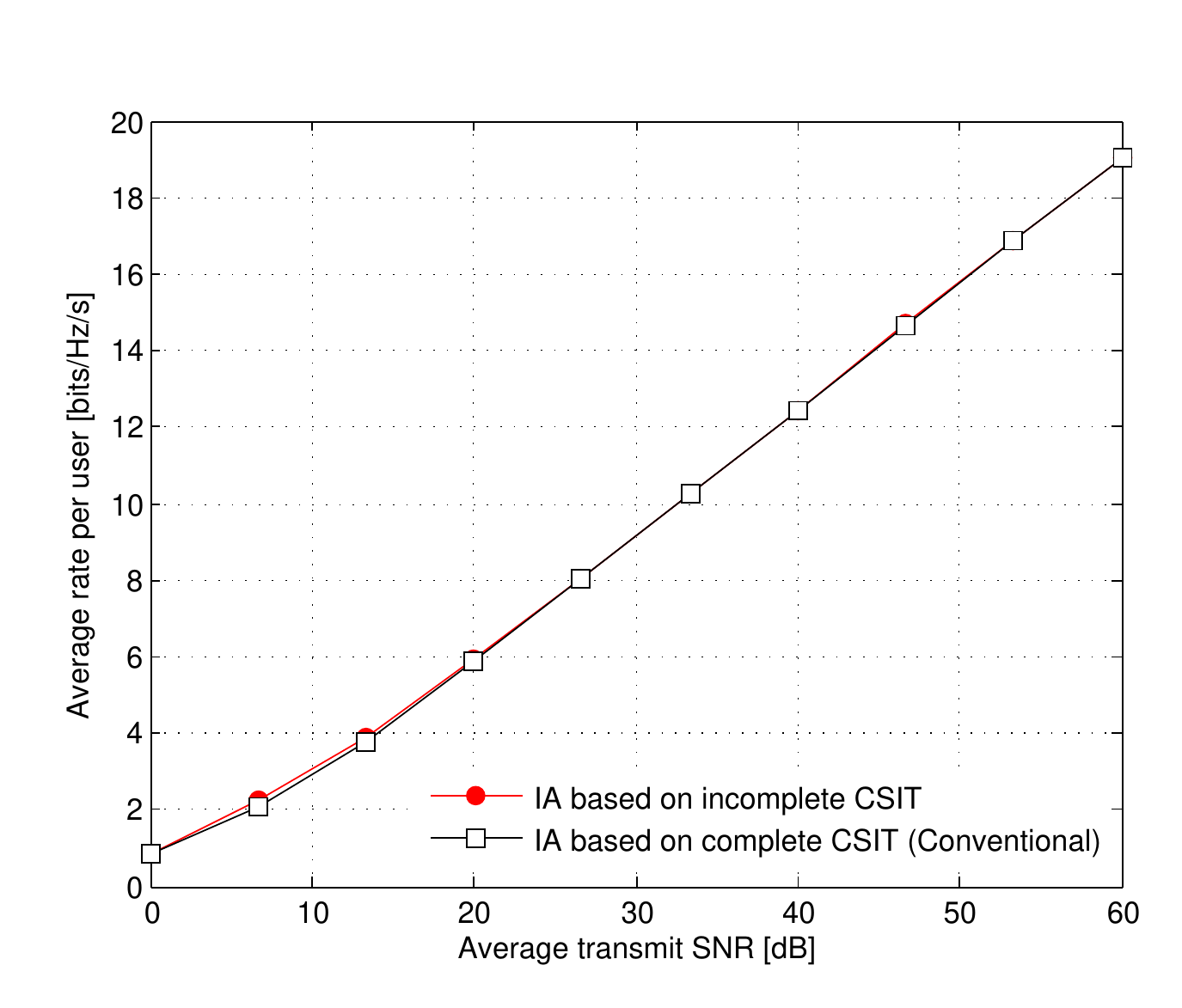}
\caption{Average rate per user in terms of the normalized TX power for the tightly-feasible $[(2,3).(2,4).(3,5).(3,2).(4,2)]$~IC.}
\label{Rate_K5}
\end{figure}
We show in Fig.~\ref{Rate_K5} the average rate per user achieved in terms of the SNR. We compare then our IA algorithm based on incomplete CSIT to the min-leakage IA algorithm based on complete CSIT (See Appendix~\ref{app:minLeak}). Our algorithm achieves virtually the same performance as the min-leakage algorithm. Hence, the reduction of $60\%$ of the feedback size (Cf. Subsection~\ref{se:Tight:Ex}) comes for ``free", making it especially interesting in practice. 



\subsection{Performance Evaluation of the CSIT allocation Algorithm}\label{se:CSI:algo} 

\begin{figure}
\centering
\includegraphics[width=1\columnwidth]{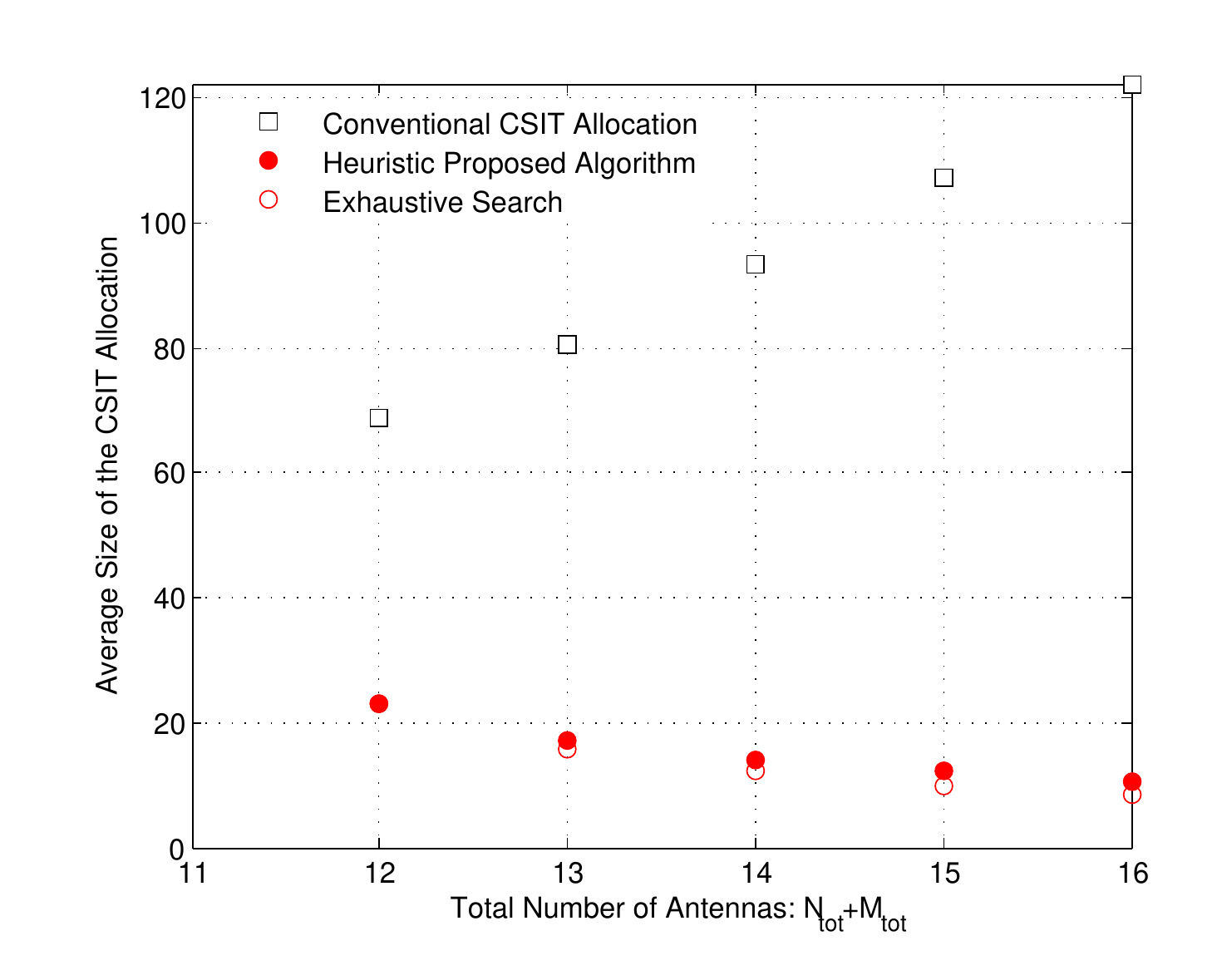}
\caption{Average CSIT allocation size in terms of the number of antennas distributed across the TXs and the RXs for $K=3$~users.}
\label{Feedback_size_K3}
\end{figure}

We will now evaluate the feedback reduction obtained with our CSIT allocation policy in super-feasible settings. Since this gain depends on the antenna configuration, we show in Fig.~\ref{Feedback_size_K3} the average size of the CSIT allocation for $K=3$~users when the antennas are allocated at random to the TXs and the RXs according to the uniform distribution. Note that the antenna configurations obtained can make IA unfeasible. When this is the case, we redistribute the antennas until a feasible antenna configuration is obtained. 

We average over $1000$~antenna configurations and the proposed heuristic CSIT allocation policy is compared with the exhaustive search. The exhaustive search consists in testing all the possibilities for removing the additional antennas\footnote{Note that a true exhaustive search through all the possible CSIT allocations (i.e., coming back to the original optimization problem \eqref{eq:SM_18}) is too complex even for trivial antenna configurations.}. For reference, we also show the average size of the complete (conventional) CSIT allocation. We consider only $K=3$~users because of the exponential complexity of the exhaustive search.


If the aggregate number of antennas is strictly smaller than~$K(K+1)=12$, it follows from Theorem~\ref{thm_feasibility} that IA cannot be feasible. The result obtained with $12$ antennas distributed between all the nodes corresponds then to the tightly-feasible case. It is hence possible to observe how our CSIT allocation algorithm leads to a significant reduction of the required CSIT without requiring any additional antenna. If more than $12$~antennas are available, each additional antenna is exploited by the heuristic algorithm to reduce the size of the CSIT allocation. This algorithm brings a reduction of the CSIT size which is only slightly smaller than the reduction brought by exhaustive search, but has a polynomial complexity. Note that allocating the antennas uniformly at random represents a worst case as it tends to generate homogeneous settings as the number of antennas increases.


\section{Discussion}
IA feasibility is studied in the literature under the assumption of full CSIT sharing. In contrast, the relation between IA feasibility and CSIT allocation is investigated in this work. Specifically, it is shown how IA can be achieved in some cases without full CSIT sharing. When extra-antennas are available, the existence of a trade-off between the number of antennas available and the CSIT sharing requirements is shown. Our approach brings a significant reduction of the feedback size while introducing no losses in terms of DoF compared to the conventional IA algorithm with full CSIT sharing. 

Furthermore, IA with incomplete CSIT sharing raises additional interesting open problems that go beyond the scope of this paper. Firstly, proving the minimality of our reduced CSIT allocation (or finding the minimal CSIT allocation policy) could not be achieved due to the difficulty in deriving a lower bound for the minimal size of a CSIT allocation preserving IA feasibility. Another interesting problem is to extend the study to multiple streams transmissions. Verifying IA feasibility represents already a difficult problem in this case so that the derivation of analytical results is challenging. However, the main idea behind our CSIT allocation algorithm directly extends to the general case with multiple-stream transmissions.

Finally, the analysis has been carried out by considering the DoF which models the performance at asymptotically high SNR. At low to medium SNR, it is expected that CSIT incompleteness will lead to some rate loss as beamforming capabilities are reduced. Thus, an interesting problem lies in the trade-off between CSIT sharing reduction and finite SNR rate performance. 

\appendix

\subsection{Minimum Leakage Interference Algorithm}\label{app:minLeak}

Many IA algorithms are already available in the literature~\cite{Gomadam2008,Schmidt2009,Kumar2010,Santamaria2010,Papailiopoulos2010,Peters2012} and each of them aims at maximizing the performance at finite SNR while converging to an IA solution at high SNR. The aim of this work being to study the feasibility of IA and not to improve on the performance of IA algorithm at finite SNR, we will use for the simulations the \emph{minimum (min-) leakage algorithm} from \cite{Gomadam2008}. It has the advantage of not requiring the knowledge of the direct channel but only the CSI required for fulfilling the IA constraints, i.e., the interfering channels.

The min-leakage algorithm can be described in our setting as follows. In a $K$-user IC, the algorithm minimizes the sum of the interference power created at the RXs which is called~$\mathrm{I}_{\IA}$ and is equal to
\begin{equation}
\mathrm{I}_{\IA}\triangleq\sum_{i=1}^K\sum_{k=1,k\neq i}^K |\bm{g}_i^{\He}\mathbf{H}_{ik}\bm{t}_k|^2.
\label{eq:ML_1}
\end{equation}
The algorithm is based on an alternating minimization in which the TX beamformers are first obtained from the RX beamformers as
\begin{equation}
\bm{t}_k=\text{eig}_{\min}\left(\sum_{i=1,i\neq k}^K \mathbf{H}_{ik}^{\He}\bm{g}_i\bm{g}_i^{\He}\mathbf{H}_{ik}\right),\qquad \forall k\in \mathcal{K}.
\label{eq:ML_2}
\end{equation}
Similarly, the RX beamformers at all RXs are then obtained from the TX beamformers as
\begin{equation}
\bm{g}_k=\text{eig}_{\min}\left(\sum_{i=1,i\neq k}^K \mathbf{H}_{ki}\bm{t}_i\bm{t}_i^{\He}\mathbf{H}_{ki}^{\He}\right),\qquad \forall k\in \mathcal{K}.
\label{eq:ML_3}
\end{equation}
The beamformers are updated iteratively until convergence to a local minimizer of~$\mathrm{I}_{\IA}$. 

\subsection{Proof of Theorem~\ref{thm_feasibility}}\label{app:feasibility} 
For $\mathcal{I}\subseteq \mathcal{J}=\{(i,j)|1\leq i,j\leq K, i\neq j\}$, we define the sets
\begin{equation}
~\begin{aligned}
\mathcal{S}_{\TX}(\mathcal{I})&\triangleq \{j|\exists k', (k',j)\in \mathcal{I}\} ,\\
\mathcal{S}_{\RX}(\mathcal{I})&\triangleq \{k|\exists j', (k,j')\in \mathcal{I}\}.
\end{aligned}
\label{eq:SM_6}
\end{equation}
Hence, $\mathcal{S}_{\RX}(\mathcal{I})$ and $\mathcal{S}_{\TX}(\mathcal{I})$ contain respectively the set of RXs and the set of TXs appearing in at least one equation of the set of equations~$\mathcal{I}$. With these notations, equation \eqref{eq:theorem_feasibility} can be rewritten as
\begin{equation}
|\mathcal{I}|\!\leq  \sum_{k\in \mathcal{S}_{\TX}(\mathcal{I})}\!\!(M_k-1) +\!\sum_{j\in\mathcal{S}_{\RX}(\mathcal{I})}\!\!(N_j-1),\qquad\forall \mathcal{I}\subseteq \mathcal{J}.
\label{eq:SM_7}
\end{equation}
Adding equations to~$\mathcal{I}$ without increasing~$\mathcal{S}_{\RX}(\mathcal{I})$ or~$\mathcal{S}_{\TX}(\mathcal{I})$ makes condition \eqref{eq:SM_7} tighter. Hence, it is only necessary to satisfy \eqref{eq:SM_7} for the sets of equations made of all the equations generated by the RXs in~$\mathcal{S}_{\RX}(\mathcal{I})$ and the TXs in~$\mathcal{S}_{\TX}(\mathcal{I})$, which is exactly the result of the theorem.

\subsection{Proof of Theorem~\ref{thm_incomplete}}\label{app:thm_incomplete} 
\begin{proof}
We have by assumption that
\begin{equation}
\mathcal{N}_{\Var}(\mathcal{S}_{\RX},\mathcal{S}_{\TX})=\mathcal{N}_{\Eq}(\mathcal{S}_{\RX},\mathcal{S}_{\TX}).
\label{eq:App_proof_1}
\end{equation}
This setting being tightly-feasible, it is possible to align interference inside this sub-IC. In the following we assume that the beamformers of the TXs and the RXs inside this sub-IC have been fixed and fulfill all IA constraints inside this sub-IC. In terms of IA feasibility, this is equivalent to replacing the TXs and the RXs inside this sub-IC by non-interfering single-antenna nodes. Indeed, these nodes do not have any ZF capabilities left but do not create any IA constraints among themselves.

We will now show that IA remains feasible in the IC obtained once these beamformers have been fixed. Since the initial IC was feasible, it holds for any subset of TX~$\mathcal{S}_{\TX}'$ and any subset of RX~$\mathcal{S}_{\RX}'$ that
\begin{equation}
\mathcal{N}_{\Var}(\mathcal{S}_{\RX}',\mathcal{S}_{\TX}')\geq \mathcal{N}_{\Eq}(\mathcal{S}_{\RX}',\mathcal{S}_{\TX}').
\label{eq:App_proof_2}
\end{equation}
A fortiori, it holds for $\mathcal{S}_{\RX}\cap\mathcal{S}_{\RX}'=\emptyset$ and $\mathcal{S}_{\TX}\cap\mathcal{S}_{\TX}'=\emptyset$ that
\begin{equation}
\mathcal{N}_{\Var}(\mathcal{S}_{\RX}\cup\mathcal{S}_{\RX}',\mathcal{S}_{\TX}\cup\mathcal{S}_{\TX}')\geq \mathcal{N}_{\Eq}(\mathcal{S}_{\RX}\cup\mathcal{S}_{\RX}',\mathcal{S}_{\TX}\cup\mathcal{S}_{\TX}').
\label{eq:App_proof_3}
\end{equation}
It follows easily from the definition of $\mathcal{N}_{\Var}$ and $\mathcal{N}_{\Eq}$ in \eqref{eq:SM_8}, that~$\forall \mathcal{A},\mathcal{A}',\mathcal{B},\mathcal{B}'\subset \mathcal{K}$, with $\mathcal{A}'\cap \mathcal{A}=\emptyset$ and $\mathcal{B}'\cap \mathcal{B}=\emptyset$, it holds that
\begin{equation}
\begin{aligned}
&\mathcal{N}_{\Var}(\mathcal{A}\cup \mathcal{A}',\mathcal{B}\cup \mathcal{B}')=\mathcal{N}_{\Var}(\mathcal{A},\mathcal{B})+\mathcal{N}_{\Var}(\mathcal{A}', \mathcal{B}')\\
&\mathcal{N}_{\Eq}(\mathcal{A}\cup \mathcal{A}',\mathcal{B}\cup \mathcal{B}')=\mathcal{N}_{\Eq}(\mathcal{A},\mathcal{B})+\mathcal{N}_{\Eq}(\mathcal{A}', \mathcal{B})\\
&~~~~~~~~~~~~+\mathcal{N}_{\Eq}(\mathcal{A}, \mathcal{B}')+\mathcal{N}_{\Eq}(\mathcal{A}', \mathcal{B}').
\end{aligned}
\label{eq:App_proof_4}
\end{equation}
Applying the relations in \eqref{eq:App_proof_4} to rewrite \eqref{eq:App_proof_3} and using also \eqref{eq:App_proof_2} gives
\begin{equation}
\begin{aligned}
\mathcal{N}_{\Var}(\mathcal{S}_{\RX}',\mathcal{S}_{\TX}')&\geq \mathcal{N}_{\Eq}(\mathcal{S}_{\RX},\mathcal{S}_{\TX}')\\
&\!\!+\mathcal{N}_{\Eq}(\mathcal{S}_{\RX}',\mathcal{S}_{\TX})+\mathcal{N}_{\Eq}(\mathcal{S}_{\RX}', \mathcal{S}_{\TX}').
\end{aligned}
\label{eq:App_proof_5}
\end{equation}
The relation \eqref{eq:App_proof_5} describes exactly all the feasibility conditions in the IC obtained once the beamformers inside the sub-IC containing the RXs in $\mathcal{S}_{\RX}$ and the TXs in $\mathcal{S}_{\TX}$ have been fixed. This shows that IA remains feasible and concludes the proof.
\end{proof}

\subsection{Proof of Theorem~\ref{thm_Min}}\label{app:thm_Min}
\begin{proof}
Let us consider w.l.o.g. the precoding at TX~$j$. By construction, TX~$j$ is allocated with the CSI relative to the sub-IC formed by the pair of sets~$(\mathcal{S}_{\RX}^{(j)},\mathcal{S}_{\TX}^{(j)})$, which is \emph{tightly-feasible}. We have shown in Appendix~\ref{app:thm_incomplete} that setting the beamformers in a tightly-feasible sub-IC to align interference in this sub-IC, does not reduce the feasibility of IA in the full IC. Thus, if all the TXs included in~$\mathcal{S}_{\TX}^{(j)}$ would design jointly their beamformers with the other TXs adapting to these TX beamformers, IA feasibility would then be preserved. Yet, all the TXs in~$\mathcal{S}_{\TX}^{(j)}$ do not necessarily share the same CSIT and thereby cannot necessarily design jointly the beamformers. Thus, it remains to prove that all the TXs included in~$\mathcal{S}_{\TX}^{(j)}$ design their beamformers in such a way that IA is achieved inside this sub-IC.

By inspection of the CSIT allocation algorithm, the CSIT allocations of all the TXs contained in~$\mathcal{S}_{\TX}^{(j)}$ are included in the CSIT of TX~$j$. Thus, TX~$j$ reproduces the precoding which is done at these TXs to obtain their TX beamformers and then align over the interference subspaces generated. This ensures the coherency between the beamformers of all the TXs in~$\mathcal{S}_{\TX}^{(j)}$ so that IA is achieved.
\end{proof}

\bibliographystyle{IEEEtran}

  \begin{IEEEbiography}[{\includegraphics[scale=0.5]{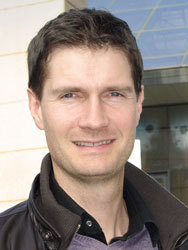}}]{David Gesbert} (IEEE Fellow) is Professor and Head of the Mobile Communications Department, EURECOM, France, where he also heads the Communications Theory Group. He obtained the Ph.D degree from Ecole Nationale Superieure des Telecommunications, France, in 1997. From 1997 to 1999 he has been with the Information Systems Laboratory, Stanford University. In 1999, he was a founding engineer of Iospan Wireless Inc, San Jose, Ca.,a startup company pioneering MIMO-OFDM (now Intel). Between 2001 and 2003 he has been with the Department of Informatics, University of Oslo as an adjunct professor. D. Gesbert has published over 200 papers and several patents all in the area of signal processing, communications, and wireless networks.

 D. Gesbert was a co-editor of several special issues on wireless networks and communications theory, for JSAC (2003, 2007, 2009), EURASIP Journal on Applied Signal Processing (2004, 2007), Wireless Communications Magazine (2006). He served on the IEEE Signal Processing for Communications Technical Committee, 2003-2008.  He was an associate editor for IEEE Transactions on Wireless Communications and the EURASIP Journal on Wireless Communications and Networking. He authored or co-authored papers winning the 2012 SPS Signal Processing Magazine Best Paper Award, 2004 IEEE Best Tutorial Paper Award (Communications Society), 2005 Young Author Best Paper Award for Signal Proc. Society journals, and paper awards at conferences 2011 IEEE SPAWC, 2004 ACM MSWiM workshop. He co-authored the book Space time wireless communications: From parameter estimation to MIMO systems, Cambridge Press, 2006.  In 2013, he was a General Chair for the IEEE Communications Theory Workshop, and a Technical Program Chair for the Communications Theory Symposium of ICC2013. He is a Technical Program Chair for IEEE ICC 2017, to be held in Paris.

\end{IEEEbiography}
 \smallskip{} 
  \begin{IEEEbiography}[{\includegraphics[scale=0.4]{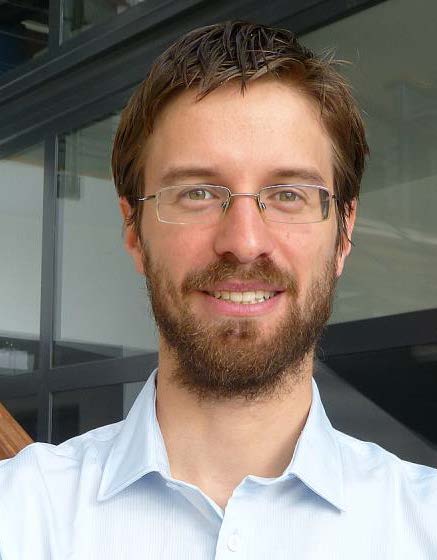}}]{Paul de Kerret} (IEEE Student Member) graduated in 2009 from Ecole Nationale Superieure des Telecommunications de Bretagne, France and obtained a diploma degree in electrical engineering from Munich University of Technology (TUM), Germany. He also earned a four year degree in mathematics at the Universite de Bretagne Occidentale, France in 2008. From January 2010 to september 2010, he has been a research assistant at the Institute for Theoretical Information Technology, RWTH Aachen University, Germany. In december 2013, he obtained a Ph.D. degree in the Mobile Communications Department at EURECOM, France, under the supervision of David Gesbert. He is the first author of several journal papers in prestigious journals and a magazine. With David Gesbert, he will give a tutorial in ICASSP 2014 on the challenges behind the cooperation of transmitters in wireless networks. He is interested in communication theory, information theory, and game theory.
\end{IEEEbiography}
\end{document}

%% file: Definitions.tex
\DeclareMathAlphabet{\mathbit}{OML}{cmr}{bx}{it}
\DeclareMathAlphabet{\mathsf}{OT1}{cmss}{m}{n}
\DeclareMathAlphabet{\mathTXf}{OT1}{cmss}{bx}{it}

\DeclareMathOperator{\diag}{diag}

\DeclareMathOperator*{\argmin}{argmin}

\DeclareMathOperator{\IA}{IA}
\DeclareMathOperator{\RX}{RX}
\DeclareMathOperator{\TX}{TX}

\DeclareMathOperator{\CN}{\mathcal{N}_{\mathbb{C}}}

\DeclareMathOperator{\NT}{NT}

\DeclareMathOperator{\Feas}{feas}
\DeclareMathOperator{\Var}{var} 
\DeclareMathOperator{\Comp}{comp}
\DeclareMathOperator{\Eq}{eq}  
\DeclareMathOperator{\Eff}{eff}  
  
\DeclareMathOperator{\Tot}{tot}  
\DeclareMathOperator{\Size}{s}   
\DeclareMathOperator{\Fix}{fix}

\newcommand{\bA}{\mathbf{A}}

\newcommand{\bH}{\mathbf{H}}

\newcommand{\bt}{\bm{t}}

\newcommand{\LB}{\left(}
\newcommand{\RB}{\right)}

\newcommand{\I}{\mathbf{I}} 
\newcommand{\norm}[1]{\lVert{#1}\rVert}

\newcommand{\Fro}{{\mathrm{F}}}

\newcommand{\trans}{{\text{T}}} 
 
\newcommand{\He}{{{\mathrm{H}}}}

\theoremstyle{remark}
\newtheorem{remark}{Remark} 
\theoremstyle{example}
\newtheorem{example}{Example} 